\documentclass[letterpaper, 10 pt, conference]{ieeeconf}  

\IEEEoverridecommandlockouts                              

\title{\LARGE \bf
Correct-by-Design Control of Parametric Stochastic Systems*
}

\author{Oliver Sch\"{o}n$^{1}$, Birgit van Huijgevoort$^{2}$, Sofie Haesaert$^{2}$, and Sadegh Soudjani$^{1}$
\thanks{*This work is supported by
the NWO Veni project CODEC (18244),
the UK EPSRC New Investigator Award CodeCPS (EP/V043676/1),
and the Horizon Europe EIC project SymAware (101070802).
}
\thanks{$^{1}$Oliver Sch\"{o}n and Sadegh Soudjani are with the School of Computing, Newcastle University, Newcastle, NE4 5TG, UK.
        {\tt\small o.schoen2@ncl.ac.uk,sadegh.soudjani@ncl.ac.uk}}%
\thanks{$^{2}$Birgit van Huijgevoort and Sofie Haesaert are with the Electrical Engineering Department, TU Eindhoven, The Netherlands
        {\tt\small b.c.v.huijgevoort@tue.nl, s.haesaert@tue.nl}}%
}

\usepackage{longtable}
\usepackage{mathtools}
\usepackage{amsfonts,amsmath,amssymb}

\usepackage{amsthm}
\usepackage{booktabs}
\usepackage{mathrsfs}  
\usepackage{tikz}
\usepackage{pgfplots}
\usetikzlibrary{positioning,calc,fit,arrows}
\usepackage{tkz-fct}
\usepackage{calc}
\usetikzlibrary{intersections}
\usetikzlibrary{decorations.shapes,shapes.geometric,shadows,arrows,automata,positioning,calendar,mindmap,backgrounds,scopes,chains,er,patterns,pgfplots.groupplots}
\usetikzlibrary{calc}
\usepackage[load-configurations=version-1]{siunitx} 
\usepackage{tikz-3dplot}
\usepackage{subfig}
\usepackage{url}
\pgfmathdeclarefunction{normal}{2}{%
	\pgfmathparse{1/(#2*sqrt(2*pi))*exp(-((x-#1)^2)/(2*#2^2))}%
}
\pgfmathdeclarefunction{dnorm}{2}{%
	\pgfmathparse{1/(#2*sqrt(2*pi))*exp(-((x-#1)^2)/(2*#2^2))}%
}
\usepackage[author={Oliver}]{pdfcomment}
\usepackage{mdframed}
\newcommand{\Refi}{\mathbf{s}}

\newtheoremstyle{theoremdd}
{\topsep}
{\topsep}
{\itshape}
{0pt}
{\bfseries}
{:}
{ }
{\thmname{#1}\thmnumber{ #2}\boldmath\textbf{\thmnote{ (#3)}}} 
\theoremstyle{theoremdd}

\newtheorem{problem}{Problem}
\newtheorem{definition}{Definition}
\newtheorem{theorem}{Theorem}
\newtheorem{proposition}{Proposition}
\newenvironment{prob}
{\begin{mdframed}
		\begin{problem}}
		{\end{problem}\vspace{.4em}\end{mdframed}}
       
\newtheorem{remark}{Remark}
\newtheorem{example}{Example}

\newcommand{\R}{\mathcal R} 

\newcommand{\norm}[1]{\left\lVert#1\right\rVert} 
\newcommand{\cdf}[1]{\mathrm{cdf}\left( #1 \right)}
\newcommand{\offset}{\gamma}

\newcommand{\T}{^\top}
\newcommand{\dxp}{dx^+}
\newcommand{\dxhp}{d\hat x^+}
\newcommand{\xp}{x^+}
\newcommand{\xhp}{\hat x^+}

\ifdefined\U
\renewcommand{\U}{\mathbb{U}}
\else
\newcommand{\U}{\mathbb{U}}
\fi

\newcommand{\Uh}{\widehat{\mathbb{U}}}
\newcommand{\given}{\;|\;}

\newcommand{\fdel}[1]{{f}_\delta(#1)}
\newcommand{\fhdel}[1]{\hat{{f}}_\delta (#1)}

\newcommand{\Z}{\mathbb{Z}}

\newcommand{\absstwo}{{\Mh\rightarrow\Mt}}

\newcommand{\satisfies}{\vDash}

\newcommand{\Tr}{\mathbf{t}}
\newcommand{\Trh}{\hat{\mathbf{t}}}

\newcommand{\X}{\mathbb{X}}
\newcommand{\Xh}{\hat{\mathbb{X}}}
\newcommand{\x}[1]{{x}_{#1}}
\newcommand{\xs}{\mathbf{x}}
\newcommand{\xin}{ {x}_0}
\newcommand{\M}{\mathbf M}
\newcommand{\Mh}{\widehat{\mathbf M}}
\newcommand{\Mt}{\widetilde{\mathbf M}}

\newcommand{\A}{\mathbb{U}}
\newcommand{\Ah}{\hat{\mathbb{U}}}
\newcommand{\y}[1]{y_{#1}}

\newcommand{\ac}[1]{u_{#1}}
\newcommand{\acs}{\mathbf{u}}
\newcommand{\rel}{\mathcal{R}}

\newcommand{\Hist}{\mathbb{H}}

\newcommand{\CA}[1]{\mathcal{#1}}

\newcommand{\ind}{\mathbf 1}
\newcommand{\AP}{\mathsf{AP}}

\newcommand{\notltl}{\neg}
\newcommand{\andltl}{\wedge}
\newcommand{\orltl}{\vee}

\newcommand{\Next}{\ensuremath{\bigcirc}}

\newcommand{\Event}{\ensuremath{\ \diamondsuit\ }}
\newcommand{\Until}{\ \CA{U}\ }

\newcommand{\alphabeth}{\Sigma}
\newcommand{\word}{\boldsymbol{\pi}}

\newcommand{\letter}{l}
\newcommand{\True}{\operatorname{\mathsf{true}}}

{} 

\newcommand{\meas}{\nu}     
\newcommand{\po}{p}     
\newcommand{\pok}{\mathbf{p}}     
\newcommand{\pk}[1]{\pok\left(#1\right)}     
\newcommand{\borel}[1]{\mathcal{B}\left(#1\right)}
\newcommand{\eps}{\varepsilon}
\newcommand{\InF}{\mathbf{i}} 

\newcommand{\Ca}{{\mathbf{C}}}
\newcommand{\Cah}{\widehat{\mathbf{C}}}
\newcommand{\Y}{\mathbb{Y}}
\newcommand{\support}{\operatorname{supp}}

\newcommand{\W}{v}
\newcommand{\Wt}{\boldsymbol{v}}

\newcommand{\Vb}{V} 

\newcommand{\trans}{\tau}
\newcommand{\Lim}{\mathbf L}

\newcommand{\xh}[1]{\hat{x}_{#1}}

\newcommand{\uh}[1]{\hat{u}_{#1}}

\newcommand{\ach}[1]{\hat{u}_{#1}}

\newcommand{\proj}{\Refi}

\renewcommand{\P}{\mathbb{P}}

\usepackage{tcolorbox}
\usepackage{mathtools}

\newtcbox{\blueb}{nobeforeafter,tcbox raise base,boxrule=0.4pt,top=0mm,bottom=0mm,
	right=0mm,left=0mm,arc=1pt,boxsep=2pt,before upper={\vphantom{dlg}},
	colframe=blue!50!black,coltext=black!25!black,colback=blue!10!white}
\newtcbox{\redb}{nobeforeafter,tcbox raise base,boxrule=0.4pt,top=0mm,bottom=0mm,
	right=0mm,left=0mm,arc=1pt,boxsep=2pt,before upper={\vphantom{dlg}},
	colframe=red!50!black,coltext=black!25!black,colback=red!10!white}

\newtcbox{\bluebs}{nobeforeafter,tcbox raise base,boxrule=0.4pt,top=0mm,bottom=0mm,
	right=0mm,left=0mm,arc=1pt,boxsep=.5pt,before upper={\vphantom{dlg}},
	colframe=blue!50!black,coltext=black!25!black,colback=blue!10!white}
\newtcbox{\redbs}{nobeforeafter,tcbox raise base,boxrule=0.4pt,top=0mm,bottom=0mm,
	right=0mm,left=0mm,arc=1pt,boxsep=.5pt,before upper={\vphantom{dlg}},
	colframe=red!50!black,coltext=black!25!black,colback=red!10!white}

\robustify{\bluebs}
\robustify{\redbs}
\robustify{\blueb}
\robustify{\redb}

\newcommand{\red}[1]{{\color{red} #1}}

\usepackage{xcolor}
\definecolor{lightblue}{rgb}{0.67, 0.9, 0.93}
\definecolor{lightgreen}{rgb}{0.67, 0.88, 0.69}
\definecolor{lightpink}{rgb}{1.0, 0.72, 0.77}
\definecolor{lightpurple}{rgb}{0.96, 0.73, 1.0}
\definecolor{lightyellow}{rgb}{0.98, 0.93, 0.37}

\usepackage{hyperref}
\hypersetup{
	colorlinks   = true,    
	urlcolor     = black,    
	linkcolor    = black,    
	citecolor    = black      
}

\newcommand{\new}[1]{{\color{blue!80!black} #1}}

\newcommand{\Sadegh}{\textcolor{magenta}}
\definecolor{cadmiumgreen}{rgb}{0.0, 0.42, 0.24}

\usepackage{tikz}

\begin{document}

\maketitle
\thispagestyle{empty}
\pagestyle{empty}


\begin{abstract}
This paper addresses the problem of computing controllers that are correct by design for safety-critical systems and can provably satisfy (complex) functional requirements. We develop new methods for models of systems subject to both stochastic and parametric uncertainties. We provide for the first time novel simulation relations for enabling correct-by-design control refinement, that are founded on coupling uncertainties of stochastic systems via sub-probability measures. Such new relations are essential for constructing abstract models that are related to not only one model but to a set of parameterized models. We provide theoretical results for establishing this new class of relations and the associated closeness guarantees for both linear and nonlinear parametric systems with additive Gaussian uncertainty. The results are demonstrated on a linear model and the nonlinear model of the Van der Pol Oscillator.
\end{abstract}




\section{INTRODUCTION}
Engineered systems in safety-critical applications are required to satisfy complex specifications to ensure safe autonomy of the system.
Examples of such application domains include autonomous cars, smart grids, robotic systems, and medical monitoring devices. It is challenging to design control software embedded in these systems with guarantees on the satisfaction of the specifications. This is mainly due to the fact that most safety-critical systems operate in an uncertain environment (i.e., their state evolution is subject to uncertainty) and that their state is comprised of both continuous and discrete variables.

Synthesis of controllers for systems on continuous and hybrid spaces generally does not grant analytical or closed-form solutions even when an exact model of the system is known.
A promising direction for formal synthesis of controllers w.r.t. high-level requirements is to use formal abstractions \cite{belta2017formal,tabuada09}.
The abstract models are built using model order reduction and space discretizations and are better suited for formal verification and control synthesis due to their finite space being amenable to exact, efficient, symbolic computational methods \cite{BK08,majumdar2020symbolic,SA13}.
Controllers designed on these finite-state abstractions can be refined to the respective original models by leveraging (approximate) similarity relations and control refinements \cite{haesaert2020robust}.

In the past two decades, formal controller synthesis for stochastic systems has witnessed a growing interest.
The survey paper \cite{lavaei2021automated} provides an overview of the current state of the art in this line of research.
Unfortunately, most of the available results require prior knowledge of the exact stochastic model of the system, which means any guarantee on the correctness of the closed-loop system only holds for that specific model. With the ever-increasing use of data-driven modeling and systems with learning-enabled components we are able to construct accurate parameterized models with the associated confidence bounds. This includes confidence sets that capture the model uncertainty via either Bayesian inference or frequentist approaches. The former includes representations of confidence sets via Bayesian system identification \cite{prando2016classical}, while the latter include non-asymptotic confidence set computations \cite{campi2005guaranteed,khorasani2018non,kieffer2013guaranteed}.

Although existing results on similarity quantification can be extended to relate a model with a model set \cite{haesaert2020robust}, these results would lead to controllers parameterized in a similar fashion as the model set.
To design a single controller that works homogeneously for all models in the model set, a new type of simulation relation needs to be developed.
In this paper, we provide an abstraction-based approach that is suitable for stochastic systems with model parametric uncertainty.
We assume that we are given an \emph{uncertainty set} that contains the true parameters, and design a controller such that the controlled system satisfies a given temporal specification uniformly on this uncertainty set without knowing the true parameters. 

\begin{prob}\label{prob:prob1}
	Can we design a controller such that the controlled system satisfies a given temporal logic specification with at least probability $p$ if the unknown true model belongs to a given set of models?
\end{prob}

The main contribution of this paper is to provide an answer to this question for the class of parameterized discrete-time stochastic systems and the class of syntactically co-safe linear temporal logic (scLTL) specifications. We define a novel simulation relation between a class of models and an abstract model, which is founded on coupling uncertainties in stochastic systems via sub-probability measures. We provide theoretical results for establishing this new relation and the associated closeness guarantees for both linear and nonlinear parametric systems with additive Gaussian uncertainty.

\smallskip

The rest of the paper is organized as follows.
After reviewing related work, we introduce in Sec.~\ref{sec:prelem} the necessary notions to deal with the stochasticity and uncertainty. We also give the class of models, the class of specifications, and the problem statement.
In Sec.~\ref{sec:framework}, we introduce our new notions of sub-simulation relations and control refinement that are based on partial coupling. We also show how to design a controller and use this new relation to give lower bounds on the satisfaction probability of the specification.
In Sec.~\ref{sec:establish_relation}, we establish the relation between parametric linear and nonlinear models and their simplified abstract models.
Finally, we demonstrate the application of the proposed approach on a linear system and the nonlinear Van der Pol Oscillator in Sec.~\ref{sec:case_study}. Due to space limitations, the proofs of statements will be provided in an extended journal version.

\medskip

\noindent\textbf{Related work:}
There are several approaches developed for parametric systems with no stochastic state transitions  \cite{haesaert2017data,makdesi_data,salamati2020data}.
%
%
%
One approach to dealing with epistemic uncertainty in stochastic systems is to model it as a stochastic two-player game, where the objective of the first player is to create the best performance considering the worst-case epistemic uncertainty. The literature on solving stochastic two-player games is relatively mature for finite state systems \cite{chatterjee2016perfect,chatterjee2012survey}.
There is a limited number of papers addressing this problem for continuous-state systems. The papers \cite{majumdar2021symbolic,majumdar2020symbolic} look at satisfying temporal logic specifications on nonlinear systems utilizing mu-calculus and space discretization.
Our approach builds on the concepts presented in \cite{haesaert2017verification} to design controllers for stochastic systems with parametric uncertainty that is compatible with both model order reduction and discretization.

\section{PRELIMINARIES AND PROBLEM STATEMENT}
\label{sec:prelem}
\subsection{Preliminaries}

The following notions are used.
A measurable space is a pair $(\X,\mathcal{ F})$ with sample space $\X$ and $\sigma$-algebra $\mathcal{F}$ defined over $\X$,
 which is equipped with a topology.
 As a specific instance of $\mathcal F$, consider
 Borel measurable spaces, i.e., $(\X,\mathcal{B}(\X))$, where $\mathcal{B}(\X)$ is the Borel $\sigma$-algebra on $\X$, that is the smallest $\sigma$-algebra containing open subsets of $\X$.
In this work, we restrict our attention to Polish spaces~\cite{bogachev2007measure}. 
A positive \emph{measure} $\meas$ on $(\X,\mathcal{B}(\X))$ is a non-negative map
$\nu:\mathcal{B}(\X)\rightarrow \mathbb R_{\ge 0}$ such that for all countable collections $\{A_i\}_{i=1}^\infty$ of pairwise disjoint sets in $\mathcal{B}(\X)$
it holds that
$\nu({\bigcup_i A_i })=\sum _i \nu({A_i})$.
A positive measure $\nu$ is called a \emph{probability measure} if $\nu(\X)=1$, and is called a \emph{sub-probability measure} if $\nu(\X)\leq 1$.

A probability measure $\po$ together with the measurable space $(\X,\mathcal{B}(\X))$ defines a \emph{probability space} denoted by $(\X,\mathcal{B}(\X),\po)$ and has realizations  $x\sim \po$.
We denote the set of all probability measures for a given measurable space $(\X,\mathcal{B}(\X))$ as $\mathcal P (\X)$.
For two measurable spaces $(\X,\mathcal{B}(\X))$ and $(\Y,\mathcal{B}(\Y))$, a \emph{kernel} is a mapping $\pok: \X \times \mathcal B(\Y)\rightarrow \mathbb R_{\geq 0}$
such that $\pk{x,\cdot}:\mathcal B(\Y)\rightarrow\mathbb{R}_{\geq 0}$ is a measure for all $x\in\X$, and $\pk{\cdot, B}: \X\rightarrow \mathbb R_{\geq 0}$ is measurable for all  $B\in\mathcal B(\Y)$.
A kernel associates to each point $x\in\X$ a measure denoted by $\pk{\cdot|x}$.
We refer to $\pok$ as a (sub-)probability kernel if in addition $\pk{\cdot|x}:\mathcal B(\Y)\rightarrow [0,1]$ is a (sub-)probability measure.
The \emph{Dirac delta} measure $\delta_a:\mathcal{B}(\X)\rightarrow [0,1]$ concentrated at a point $a\in\X$ is defined as $\delta_a(A)=1$ if $a\in A$ and $\delta_a(A)=0$ otherwise, for any measurable $A$.

For given sets $A$ and $B$, a relation $\rel\subset A\times B$ is a subset of the Cartesian product $A\times B$. The relation $\rel$ relates $x\in A$ with $y\in B$ if $(x,y)\in\rel$, written equivalently as $x\rel y$.
For a given set $\Y$, a metric or distance function $\mathbf d_\Y$ is a function $\mathbf{d}_\Y: \Y\times \Y\rightarrow \mathbb R_{\ge 0}$
satisfying the following conditions for all $y_1,y_2,y_3\in\Y$:
$\mathbf d_\Y(y_1,y_2)=0$ iff $y_1=y_2$;
$\mathbf d_\Y(y_1,y_2)=\mathbf d_{\Y}(y_2,y_1)$;  and
$\mathbf d_\Y(y_1,y_3)\leq \mathbf d_\Y(y_1,y_2) +\mathbf d_\Y(y_2,y_3)$.


\subsection{Discrete-time uncertain stochastic systems}
We consider discrete-time nonlinear systems perturbed by additive stochastic noise under model-parametric uncertainty. This modeling formalism is essential if we can only access an uncertain model of this stochastic system.
Consider the set of models $\{\M(\theta)$ with $\theta\in\Theta\}$, parametrized with $\theta$:
\begin{equation}
\label{eq:model}
	\M(\theta): \left\{ \begin{array}{ll}
		x_{k+1}&= f(x_k,u_k;\theta) + w_k\\
		y_k& = h(x_k\red{}),
	\end{array} \right.
\end{equation}
where the system state, input, and observation at the $k^{\text{th}}$ time-step are denoted by $x_k, u_k, y_k$, respectively. The functions $f$ and $h$ specify, respectively, the parameterized state evolution of the system and the observation map. The additive noise is denoted by $w_k$, which is an independent, identically distributed noise sequence with distribution $w_k\sim p_w(\cdot)$.
We assume throughout this paper that the uncertain parameter $\theta$ belongs to a known, bounded polytope $\Theta\subset\mathbb{R}^{\mathfrak{p}}$ for some $\mathfrak{p}$.


%
A controller for the model \eqref{eq:model} is denoted by $\Ca$ and is implemented as depicted in Fig.~\ref{fig:control}. We denote the composition of the controller $\Ca$ with the model $\M(\theta)$ as $ \Ca\times \M(\theta)$.

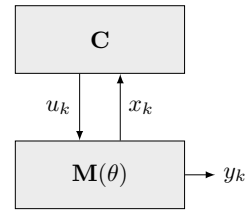
\begin{figure}[htp]
	\centering
	\scalebox{0.9}{	\begin{tikzpicture}
	
			\node[draw,fill=gray!15, rectangle, minimum width=2.5cm,node distance=3cm, minimum height=1cm] (M2) {$\M(\theta)$};
		
 			\node[above of =M2,draw,fill=gray!15, rectangle, minimum width=2.5cm,node distance=2cm, minimum height=1cm] (Ca2) {$\Ca$};
			\node[right of =M2, node distance= 2cm](out2){$\y{k}$};
			\path[draw] (Ca2.south) ++(.3cm,0)  edge[latex-,right] node{$ \x{k}$} ([xshift=.3cm]Ca2.south |- M2.north);
			\path[draw] (Ca2.south) ++(-.3cm,0)  edge[-latex,left] node{$\ac{k}$} ([xshift=-0.3cm]Ca2.south |- M2.north);
			\path[draw] (M2)  edge[-latex,left]   (out2);

		\end{tikzpicture}
	}
	\caption{Control design for a parameterized stochastic model.}
	\label{fig:control}
\end{figure}


\subsection{Problem statement}\label{sec:prob_statem}
Consider a parameterized model in \eqref{eq:model} with $\theta\in\Theta$.
We are interested in designing a controller $\Ca$ to satisfy temporal specifications $\psi$ on the output of the model. This is denoted by $\Ca\times \M(\theta) \satisfies \psi$. Since the true $\theta$ is unknown, can we design a controller that does not depend on $\theta$ and that ensures the satisfaction of $\psi$ with at least probability $p_\psi$? We formalize this problem as follows:

\begin{prob}\label{prob:prob2}
	\setlength{\belowdisplayskip}{0pt}
	 For a given specification $\psi$ and a threshold $p_\psi\in (0,1)$, design a controller $\mathbf C$ for $\M(\theta)$ that does not depend on the parameter $\theta$ and that
	 \[ \P\left(\Ca\times \M(\theta) \satisfies \psi \right)\geq   p_\psi,\quad  \forall \theta\in\Theta. \]
\end{prob}
The controller synthesis for stochastic models is studied in \cite{haesaert2020robust} through coupled simulation relations. Although these simulation relations can relate one abstract model to a set of parameterized models $\M(\theta)$, these relations would lead to a control refinement that is still dependent on the true model or true parameter. Therefore, this approach is unfit to solve Prob.~\ref{prob:prob2}, since the required true parameter $\theta$ is unknown.
As one of the main contributions of this paper, we start from a parameter-independent control refinement and compute a novel simulation relation based on a sub-probability coupling (see Sec.~\ref{sec:framework}) to synthesize a single controller for all $\theta\in\Theta$.

\subsection{Markov decision processes}\label{sec:MDPs}
Consider an abstract model $\Mh$, based on which we wish to design a controller:
\begin{equation}
\label{eq:nom_dynamics}
\Mh:
 \left\{ \begin{array}{ll}
	\hat x_{k+1} &= \hat f_{\mathfrak n}(\hat x_k,\hat u_k) + \hat w_k,\quad \hat w_k\sim \hat p_{\hat w}(\cdot),\\
\hat y_k &= \hat h_{\mathfrak n}(\hat x_k).
\end{array} \right.
%
\end{equation}
This abstract model could for example be $\M(\theta)$ but for a given valuation of parameters $\theta = \theta_{\mathfrak n}$, or any other model constructed by space reduction or discretization. The systems $\Mh$ in \eqref{eq:nom_dynamics} and $\M(\theta)$
with a given fixed $\theta$ can equivalently be described by a general Markov decision process, studied previously for formal verification and synthesis of controllers \cite{haesaert2017verification,SA13}. Given $\hat x_k, \hat u_k$, we can model the stochastic state transitions of  $\Mh$ with a probability kernel $\hat \Tr(\cdot|\hat x_k,\hat u_k)$ that is computed based on $\hat f_{\mathfrak n}$ and $\hat p_{\hat w}(\cdot)$ (similarly for $\M(\theta)$). This leads us to the representation of the systems as Markov decision processes, which is defined next.


%
%
%
\begin{definition}[{Markov decision process (MDP)}]
	An MDP is a tuple $\M=(\X,x_0,\A,\Tr )$ with 
	$\X$ the state space with states $x\in\X$; 
	$x_0\in\X$ the initial state;
	$\A$  the input space with input $u\in\A$;
	and
	$\Tr:\X\times\A\times\mathcal B(\X)\rightarrow[0,1]$, a probability kernel assigning to each state $x\in \X$ and input $u\in \A$ pair a probability measure $ \Tr(\cdot| x,u)$ on $(\X,\mathcal B(\X))$. 
\end{definition}

We indicate the input sequence of an MDP $\M$ by $\acs:= \ac{0},\ac{1},\ac{2},\ldots$
and we define its (finite) \emph{executions} as sequences of states $\xs=\x{0},\x{1},\x{2},\ldots$ (respectively, $\xs_N=\x{0},\x{1},\x{2},\ldots, \x{N}$) initialised with the initial state  $\xin$ of $\M$ at $k=0$.
In each execution, the consecutive state $x_{k+1}\in\X$
is obtained as a realization $x_{k+1}\sim\Tr\left(\cdot| x_k, u_k \right)$ of the controlled Borel-measurable stochastic kernel. 
Note that for a parametrized MDP $\M(\theta)$ its transition kernel also depends on $\theta$ denoted as $\Tr\left(\cdot| x_k, u_k; \theta \right)$.

As in Eq.~\eqref{eq:nom_dynamics}, we can assign an output mapping $h:\X\rightarrow \Y$ and a metric $\mathbf d_{\Y}$ to the MDP $\M$ to get a general MDP.
\begin{definition}[General Markov decision process (gMDP)]
	A gMDP is a tuple $\M\!=\!(\X,x_0,\A, \Tr, h,  \Y)$ that combines an MDP $(\X,x_0,\A, \Tr)$ with the output space $\Y$ and a measurable output map $h:\X\rightarrow\Y$.
A metric $\mathbf d_\Y$ decorates the output space $\Y$. 
\end{definition}
The gMDP semantics are directly inherited from those of the MDP.
Furthermore, output traces of the gMDP are obtained as mappings of (finite) MDP state executions, namely
$\mathbf y:= y_0, y_1, y_2,\ldots$ (respectively, $\mathbf y_N:= y_0, y_1, y_2,\ldots, y_N$),
where $y_k= h(x_k)$. 
%
The execution history $(x_0, u_0, x_1,\ldots, u_{N-1}, x_N)$ grows with the number of observations $N$ and takes values in the \emph{history space} $\Hist_N := (\X \times \A )^{N} \times \X$.
A control policy or controller for $\M$ is a sequence of policies mapping the current execution history to a control input.
\begin{definition}[Control policy]
	\label{def:markovpolicy}
A control policy $\boldsymbol{\mu}$ is a sequence $\boldsymbol{\mu}=(\mu_0,\mu_1,\mu_2,\ldots)$ of universally measurable maps $\mu_k:\Hist_k\rightarrow \mathcal P(\A,\mathcal B(\A))$, $k\in\mathbb N:=\{0,1,2,\ldots\}$, from the execution history to the input space.
\end{definition}
 As special types of control policies, we differentiate Markov policies and finite memory policies.
A \emph{Markov policy} $\boldsymbol{\mu}$ is a sequence $\boldsymbol{\mu}=(\mu_0,\mu_1,\mu_2,\ldots)$ of universally measurable maps $\mu_k:\X\rightarrow \mathcal P(\A,\mathcal B(\A))$, $k\in\mathbb N$, from the state space $\X$ to the input space.
We say that a Markov policy is \emph{stationary}, if $\boldsymbol{\mu}=(\mu,\mu,\mu,\ldots)$ for some $\mu$.
\emph{Finite memory policies} first map the finite state execution of the system to a finite set (memory). The input is then chosen similar to the Markov policy as a function of the system state and the memory state. This class of policies is needed for satisfying temporal specifications on the system executions.
%


In the next subsection, we formally define the class of specifications studied in this paper.

\subsection{Temporal logic specification}
Consider a set of atomic propositions $AP := \{ p_1, \ldots, p_L \}$ that defines an \emph{alphabet} $\alphabeth := 2^{AP}$, where any \emph{letter} $\letter\in\alphabeth$ is composed of a set of atomic propositions.  An infinite string of letters forms a \emph{word} $\word=\letter_0\letter_1\letter_2\ldots\in\alphabeth^{\mathbb{N}}$. We denote the suffix of $\word$ by $\word_i = \letter_i\letter_{i+1}\letter_{i+2}\ldots$ for any $i\in\mathbb N$.
Specifications imposed on the behavior of the system are defined as formulas composed of atomic propositions and operators. We consider the syntactically co-safe subset of linear-time temporal logic properties \cite{kupferman2001model} abbreviated as scLTL.
%
This subset of interest consists of temporal logic formulas constructed according to the following syntax
\begin{equation*}
	\psi ::=  p \ |\ \notltl p \ |\ \psi_1 \vee\psi_2  \ |\ \psi_1 \andltl \psi_2 \ |\ \psi_1 \Until \psi_2 \ |\ \Next \psi,
\end{equation*}   where $p\in \AP$ is an atomic proposition.
The \emph{semantics} of scLTL are defined recursively over $\word_i$ as
$\word_i \satisfies p$ iff $p \in \letter_i$;
$\word_i \satisfies \psi_1 \andltl  \psi_2  $ iff $ ( \word_i \satisfies \psi_1 ) \andltl ( \word_i \satisfies \psi_2 ) $;
$\word_i \satisfies \psi_1 \orltl  \psi_2  $ iff $ ( \word_i \satisfies \psi_1 ) \orltl ( \word_i \satisfies \psi_2 ) $;
$\word_i \satisfies  \psi_1 \Until \psi_2 $ iff $\exists j \geq i \text{ subject to } (\word_j \satisfies \psi_2 ) $ and $\word_k \satisfies \psi_1, \forall k \in \{i, \ldots j-1\}$; and
$\word_i \satisfies \Next \psi$ iff $\word_{i+1} \satisfies \psi$.
The eventually operator  $\Event \psi$ is used in the sequel as a shorthand for $\True\Until  \psi $.
We say that $\word\satisfies\psi$ iff $\word_0\satisfies\psi$.

Consider a labeling function $\mathcal L: \Y\rightarrow \Sigma$ that assigns a letter to each output. Using this labeling map, we can define temporal logic specifications over the output of the system.
Each output trace of the system $\mathbf y \!=\! y_0,y_1,y_2,\ldots$ can be translated to a word as $\word \!=\! \mathcal L(\mathbf y)$.
We say that a system satisfies the specification $\psi$ with the probability of at least $p_\psi$ if
$\mathbb P(\word\satisfies \psi) \ge p_\psi.$
When the labeling function $\mathcal L$ is known from the context, we write $\mathbb P(\M\satisfies \psi)$ to emphasize that the output traces of $\M$ are used for checking the satisfaction.

Satisfaction of scLTL specifications can be checked using their alternative representation as deterministic finite-state automata, defined next.
	\begin{definition}[Deterministic finite-state automaton (DFA)]
	\label{def:dfa}
	A DFA is a tuple $\mathcal A = \left(Q,q_0,\Sigma,F,\trans\right)$, where $Q$ is the finite set of locations of the DFA, $q_0\in Q$ is the initial location, $\Sigma$ the finite alphabet, $F\subset Q$ is the set of accepting locations, and $\trans: Q\times\Sigma\rightarrow Q$ is the transition function.
\end{definition}

For any $n\in\mathbb{N}$, a word $\boldsymbol{\omega}_n = \omega_0\omega_1\omega_2\ldots\omega_{n-1}\in\Sigma^n$ is accepted by a DFA $\mathcal A$ if there exists a finite run $q = (q_0,\ldots,q_n)\in Q^{n+1}$ such that $q_0$ is the initial location, $q_{i+1} = \trans(q_i,\omega_i)$ for all $0\leq i \le n$, and $q_n\in F$.

\smallskip

In the next section, we pave the way for answering Prob.~\ref{prob:prob2} on the class of scLTL specifications by an abstraction-based control design scheme. We provide a new similarity relation that enables parameter-independent control refinement from an abstract model to the original concrete model that belongs to a parameterized class.

\section{CONTROL REFINEMENT VIA SUB-SIMILARITY RELATIONS}
\label{sec:framework}
Consider a set of models $\{\M(\theta)$ with $\theta\in\Theta\}$ and suppose that we have chosen an abstract (nominal) model $\Mh$ based on which we wish to design a \emph{single} controller and quantify the satisfaction probability over all models $\M(\theta)$ in the set of models. In this section, we start by formalizing the notion of 
a \emph{state mapping}
and an \emph{interface function}, that together form the control refinement. Then, we investigate the conditions under which a single controller for $\Mh$ can be refined to a controller for all $\M(\theta)$ independent of the parameter $\theta$.
This leads us to the novel concept of \emph{sub-probability couplings} and simulation relations.


\subsection{Control refinement}
Consider the MDP $\Mh=(\Xh,\xh{0},\Ah,\Trh)$ as the abstract model, the MDP $\M(\theta)=(\X,x_0,\A,\Tr(\theta))$ referred to as the concrete MDP, and an abstract controller $\Cah$ for $\Mh$.
To refine the controller $\Cah$ on $\Mh$ to a controller for $\M(\theta)$, we define a pair of interfacing functions consisting of a \emph{state mapping} that translates the states $ x\in\X$  to the states $\xh{} \in\Xh$  and an \emph{interface function} that refines the inputs $\uh{}\in\Ah$ to the control inputs $u\in\A$.

\noindent{\bfseries Interface function:}
We define an interface function \cite{girard2009hierarchical,girard2011approximate,lavaei2019compositional} as a Borel measurable stochastic kernel
$\InF: \Xh \times\X\times\Ah\rightarrow\borel{\A}$ such that
any input $\ach{k}$ for $\Mh$ is refined to an input $\ac{k}$ for $\M(\theta)$ as
\begin{equation}\label{eq:input_refi}
	\ac{k}\sim \InF(\cdot\,|\xh{k}, x_k, \uh{k}).
\end{equation}

\noindent{\bfseries State mapping:} We can define the state mapping in a general form as a stochastic kernel $\Refi$ that maps the current state $\xh{k}$ and input $\ach{k}$ to the next state $\xh{k+1}$ of the abstract model.
The next state $\xh{k+1}$ has the distribution specified by $\Refi$ as
\begin{equation*}
	\xh{k+1} \sim \Refi(\cdot |\xh{k}, \x{k}, \x{k+1},\ach{k}).
\end{equation*}
This state mapping is coupled with the concrete model via its states $\x{k}, \x{k+1}$ and depends implicitly on $\ac{k}$ through Eq.~\eqref{eq:input_refi}.

Fig.~\ref{fig:abstract_control} illustrates how the abstract controller $\Cah$ defines a control input $\uh{k}$ as a function of the abstract state $\xh{k}$. Based on the state mapping $\Refi$ and the interface function $\InF$, the abstract controller $\Cah$ can be refined to a controller for $\M(\theta)$ as depicted in Fig.~\ref{fig:contrl_ref}.
In this figure, the states $\x{k}$ from the concrete model $\M(\theta)$ are mapped to the abstract states $\xh{k}$ with $\Refi$, the control inputs $\ach{k}$ are obtained using $\Cah$ and $\xh{k}$, and these inputs are then refined to control inputs $\ac{k}$ for $\M(\theta)$ using the interface function $\InF$.

\usetikzlibrary{math}
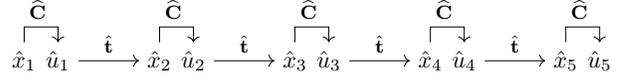
\begin{figure}[htp]
	\centering
	\scalebox{0.9}{	\begin{tikzpicture}
\foreach \i in {5, 4, ..., 1}{
	\node[name=xa\i] at (\i/5*10,0) {$\hat x_\i$};
}
\foreach \i in {5, 4, ..., 1}{
	\node[name=aa\i] at (\i/5*10+.5,0) {$\ach{\i}$};
}

\foreach \i in {5, 4, ..., 1}{
	\path[draw,->] (xa\i) -- +(0,.5)-| node[left,above,xshift = -0.3cm]{\footnotesize$\widehat{\Ca}$} (aa\i);
}
\foreach \i in {4, 3, ..., 1}{
\tikzmath{int \j;\j=\i+1;}
	\path[draw,->] (aa\i)-- node[above ]{\footnotesize$\Trh $} (xa\j);
}

		\end{tikzpicture}
	}
		\caption{Controller $\Cah$ on the abstract model $\Mh$.}
		\label{fig:abstract_control}
\end{figure}

\usetikzlibrary{}
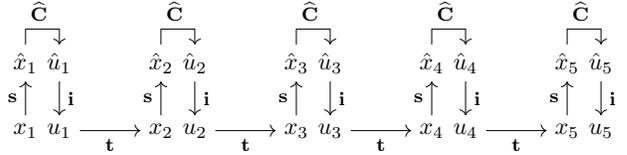
\begin{figure}[htp]
	\centering
	\scalebox{0.9}{	\begin{tikzpicture}
			\foreach \i in {5, 4, ..., 1}{
				\node[name=xa\i] at (\i/5*10,0) {$\hat x_\i$};
			}
			\foreach \i in {5, 4, ..., 1}{
				\node[name=aa\i] at (\i/5*10+.5,0) {$\ach{\i}$};
			}
			\foreach \i in {5, 4, ..., 1}{
				\node[name=x\i] at (\i/5*10,-1) {$ x_\i$};
			}
			\foreach \i in {5, 4, ..., 1}{
				\node[name=a\i] at (\i/5*10+.5,-1) {$\ac{\i}$};
			}

			\foreach \i in {5, 4, ..., 1}{
				\path[draw,->] (xa\i) -- +(0,.5)-| node[left,above,xshift = -0.3cm]{\footnotesize$\widehat{\Ca}$} (aa\i);
			}
			\foreach \i in {4, 3, ..., 1}{
				\tikzmath{int \j;\j=\i+1;}
				\path[draw,->] (a\i)-- node[below]{\footnotesize${\Tr}$
				} (x\j);
			}

			\foreach \i in {5, 4, ..., 1}{
				\path[draw,->] (x\i) -- node[left]{\footnotesize$\proj$
				} (xa\i);
			}
			\foreach \i in {5, 4, ..., 1}{
				\path[draw,->] (aa\i) -- node[right]{\footnotesize${\InF}$
				} (a\i);
			}
		\end{tikzpicture}
	}
	\caption{Control refinement that uses the abstract controller $\Cah$ to control the concrete model $\M(\theta)$. }\label{fig:contrl_ref}
\end{figure}


We now question under which conditions on the interface function and the models the refinement is valid and preserves the satisfaction probability. This is addressed in the next subsection.

\subsection{Valid control refinement and sub-simulation relations}\label{sec:partialStoch}

Before diving into the definition of a valid control refinement that is also amenable to models with parametric uncertainty, we introduce a relaxed version of approximate simulation relations   \cite{haesaert2017verification} based on sub-probability couplings.

We define a sub-probability coupling with respect to a given relation $\mathcal R\subset \Xh\times\X$ as follows.

\begin{definition}[Sub-probability coupling]
\label{def:submeasure_lifting}
	Given $\po \in\mathcal P(\X)$, $\hat\po\in \mathcal P(\Xh)$,  $\mathcal R\subset\Xh\times\X$, and a value $\delta\in[0,1]$, we say that a sub-probability measure $\W$ over $(\Xh\times\X , \mathcal B(\Xh\times\X))$ with $\W(\Xh\times\X)\geq 1-\delta$ is a \emph{sub-probability coupling} of $\hat\po$ and $\po$ over $\rel$ if
	\begin{itemize}
		\item[a)] $\W(\Xh\times\X)=\W(\R)$, that is, the probability mass of $\W$ is located on $\rel$;
		\item[b)]   for all measurable sets $A\subset \Xh$, it holds that
		$\W(A\times \X)\leq\hat\po(A)$;
		and
		\item[c)]  for all measurable sets $A\subset \X$, it holds that $\W(\Xh\times A)\leq\po(A)$.
	\end{itemize} 
\end{definition}
Note that condition a) of Def.~\ref{def:submeasure_lifting} implies $\W(\rel)\geq1-\delta$.

\begin{remark}
	For the parametrized case, i.e., $\hat\po(\cdot)$ and $\po(\cdot|\theta)$, the sub-probability coupling $\W(\cdot|\theta)$ may likewise depend on $\theta$.
	Furthermore, though we define a sub-probability coupling $\W$ as a probability measure over the probability spaces $\X$ and $\Xh$, we use it in its kernel form $\Wt$ in the remainder of this paper. We obtain a particular probability measure $\W$ from $\Wt$ for a fixed choice of $\hat x,x,\ach{}$.
\end{remark}

Let us now define $(\varepsilon, \delta)$-sub-simulation relations for stochastic systems, where $\varepsilon$ indicates the error in the output mapping and $\delta$ indicates the closeness in the probabilistic evolution of the two systems. 
\begin{definition}[($\varepsilon,\delta$)-sub-simulation relation (SSR)]
\label{def:subsim}
	Consider two gMDPs
	$\M\!=\!(\X,x_0,\A, \Tr, h, \Y)$ and $\Mh\!=\!(\Xh,\xh{0},\Ah, \Trh, \hat h,  \Y)$, a measurable relation $\rel\subset \Xh\times\X$, and an interface function \(\InF :\Xh\times\X\times\Ah \rightarrow \A\). If there exists a sub-probability kernel  $\Wt(\cdot|\xh{},x, \uh{})$
	such that
	\begin{itemize}\itemsep=0pt
		\item[(a)] $(\hat x_{0}, x_{0})\in \rel$;
		\item[(b)] for all $(\hat x,x)\in\rel$ and $\ach{}\in \Ah$, $\Wt(\cdot|\hat x,x,\ach{})$ is a sub-probability coupling  of $\hat \Tr(\cdot|\hat x,\ach{})$ and $\Tr(\cdot|x,\InF(\xh{},\x{},\ach{}))$ over $\rel$  with respect to $\delta $ (see Def.~\ref{def:submeasure_lifting});
		\item[(c)] $\forall (\hat x,x)\in\rel: \mathbf d_\Y(\hat h(\hat x),h(x)) \leq \varepsilon$;
	\end{itemize}
then $\Mh$ is in an $(\varepsilon,\delta)$-SSR with $\M$  that is denoted as \mbox{$\Mh\preceq^{\delta}_\eps\M$}.
\end{definition}
\begin{remark}
Both Defs.~\ref{def:submeasure_lifting}-\ref{def:subsim} are technical relaxations of the definitions of $\delta$-lifting and of $(\varepsilon,\delta)$-stochastic simulation relations provided in \cite{haesaert2017verification}. These relations quantify the similarity between two models by bounding their transition probabilities with $\delta$ and their output distances with $\varepsilon$.

\end{remark}
For a model class $\M(\theta)$, where $\hat \Tr(\cdot|\hat x,\ach{})$ and $\Tr(\cdot|x,\InF(\xh{},\x{},\ach{});\theta)$, the above definition allows us to have an interface function $\InF$ that is independent of $\theta$ but a sub-probability kernel $\Wt(\cdot|\xh{},x, \uh{};\theta)$ which may depend on $\theta$.
In order to solve Prob.~\ref{prob:prob2}, we require the state mapping $\Refi(\cdot|\xh{k}, \x{k}, \x{k+1}, \ach{k})$
to also be independent of $\theta$. This leads us to a condition under which the state mapping $\Refi$ gives a valid control refinement as formalized next.

\begin{definition}[\bfseries Valid control refinement]
\label{def:validrefinement}
	Consider an interface function $\InF: \Xh \times\X\times\Ah\rightarrow\borel{\A}$ and a sub-probability kernel  $\Wt$ according to Def.~\ref{def:subsim}
	. We say that a state mapping $\Refi: \Xh\times\X\times\X\times\Ah\rightarrow\mathcal P(\Xh)$ 
	defines a \emph{valid control refinement} if the composed probability measure
	\begin{align*}
	\bar{\Refi}&(d\xh{k+1} \times d\x{k+1}|\xh{k}, \x{k}, \ach{k}):=\nonumber\\
	&\Refi(d\xh{k+1}| \xh{k}, \x{k}, x_{k+1}, \ach{k}) \Tr(d\x{k+1}|\x{k},\InF(\xh{k}, \x{k}, \ach{k}))
	\end{align*}
	 upper-bounds the sub-probability coupling $\Wt$, namely
	 \begin{equation}
	 \label{eq:sub_prob}
	\bar{\Refi}(A|\xh{k}, \x{k}, \ach{k})\geq \Wt(A|\xh{k}, \x{k}, \ach{k}),
	\end{equation}
	for all measurable sets $A\subset \Xh\times \X$, all $(\xh{k}, \x{k})\in\rel$, and all $\ach{k}\in\Ah$.
\end{definition}
Note that for a model class $\M(\theta)$ that is in relation with a model $\Mh$ the right-hand side of \eqref{eq:sub_prob} can depend on $\theta$ and the left-hand side can depend on $\theta$ only via the dynamics of $\M(\theta)$ represented by the kernel $\Tr$, i.e.,
\begin{align*}
	&\int_{(\xhp{},\xp{})\in A}\!\!\!\!\!\!\!\!\!\!\!\!\!\!\!\!\!\!\Refi(d\xhp{}| \xh{k}, \x{k}, x_{k+1}, \ach{k}) \Tr(d\xp{}|\x{k},\InF(\xh{k}, \x{k}, \ach{k});\theta)\\
	&\hspace{150pt}\geq\Wt(A|\xh{k}, \x{k}, \ach{k};\theta).
\end{align*}
The next theorem states that similar to the $(\varepsilon, \delta)$-approximate simulation relation defined in \cite{haesaert2017verification}, there always exists at least one valid control refinement for our newly defined SSR.
\begin{theorem}
\label{thm:exist_refine}
For two gMDPs $\Mh$ and $\M$ with $\Mh\preceq^{\delta}_\eps\M(\theta)$, there always exists a valid control refinement.
\end{theorem}


Note that in general there is more than one valid control refinement
for given $\Mh$ and $\M$ with $\Mh\preceq^{\delta}_\eps\M(\theta)$. This allows us to choose an interface function $\InF$ not dependent on $\theta$.
%
The above theorem states that our new framework fully recovers the results of \cite{haesaert2017verification} for non-parametric models.

Although Def.~\ref{def:validrefinement} gives a sufficient condition for a valid refinement, it is not a necessary condition. For specific control specifications, one can also use different refinement strategies such as the one presented in \cite{haesaert2018temporal}, where information on the value function and relation is used to refine the control policy. 

In the next theorem, we establish that our new similarity relation is transitive. This property is very useful when the abstract model is constructed in multiple stages of approximating the concrete model. We exploit this property in the case study section.

\begin{theorem}
\label{thm:transitive}
Suppose $\Mt\preceq^{\delta_1}_{\eps_1}\Mh$ and $\Mh\preceq^{\delta_2}_{\eps_2}\M$. Then, we have $\Mt\preceq^{\delta}_{\eps}\M$ with $\delta = \delta_1+\delta_2$ and $\eps = \eps_1+\eps_2$.
\end{theorem}

\subsection{Temporal logic control with sub-simulation relations}
For designing controllers to satisfy temporal logic properties expressed as scLTL specifications, we employ the representation of the specification as a DFA $\mathcal A = \left(Q,q_0,\Sigma,F,\trans\right)$ (see Def.~\ref{def:dfa}).
We then use a robust version of the dynamic programming characterization of the satisfaction probability \cite{haesaert2020robust} defined on the product of $\Mh$ and $\mathcal A$. We now provide the details of this characterization that encode the effects of both $\varepsilon$ and $\delta$.

Denote the $\varepsilon$-neighborhood of an element $y\in\Y$ as
\[B_\varepsilon(\hat y) :=\{y \in \Y|\,  \mathbf{d}_\Y\left(y,\hat y\right)\leq \varepsilon\}, \]
and a Markov policy $\mu:\Xh\times Q\rightarrow \mathcal P(\Ah,\mathcal B(\Ah))$.
Similar to dynamic programming with perfect model knowledge \cite{Sutton2018RL}, we utilize value functions $\Vb_l^\mu:\Xh\times Q\rightarrow [0,1]$, $l\in\{0,1,2,\ldots\}$ that represent the probability of starting in $(x_0, q_0)$ and reaching the accepting set $F$ in $l$ steps. These value functions are connected recursively via operators associated with the dynamics of the systems.

Let the initial value function $\Vb_0\equiv 0$. Define the \emph{$(\eps,\delta)$-robust operator} \( \mathbf T_{\eps,\delta}^{\mu}\), acting on value functions as
\begin{align*}
	\notag \textstyle \mathbf T_{\eps,\delta}^{\mu} (\Vb)(\hat x,q)\!:=\!\Lim&\Big(\!\!\int_{\hat\X} \min_{q'\in \bar\trans(q,\hat x')}\!\!\left\{\max\left\{1_{F}(q'), \Vb(\hat x',q')\right\}\right\}\\
	&\hspace{1cm}\times \Trh(d\hat x'|\hat x,\mu)-\delta\Big),
\end{align*}
with $\bar\trans(q,\hat x'):=\{\trans(q,\alpha)\mbox{ with }\alpha \in\mathcal L(B_\varepsilon(\hat h (\hat x')))\}$. The function $\Lim:\mathbb R\rightarrow [0,1]$ is the truncation function with $\Lim(\cdot) := \min(1,\max(0,\cdot))$, and $\ind_F$ is the indicator function of the set $F$.
The value functions are connected via $\mathbf T_{\eps,\delta}^{\mu}$ as
\begin{equation*}
	\Vb_{l+1}^\mu = \mathbf T_{\eps,\delta}^{\mu}(\Vb_l^\mu),\quad l\in\{0,1,2,\ldots\}.
\end{equation*}
Furthermore, we define the \emph{optimal $(\eps,\delta)$-robust operator}
\begin{equation*}
	\mathbf T_{\eps,\delta}^{\ast} (\Vb)(\hat x,q):=\sup_{\mu} \mathbf T_{\eps,\delta}^{\mu} (\Vb)(\hat x,q).
\end{equation*}
The outer supremum is taken over Markov policies $\mu$ on the product space $\Xh\times Q$.
Based on \cite{haesaert2020robust}, Cor.~4, we can now define the lower bound on the satisfaction probability by looking at the limit case $l\rightarrow\infty$, as given in the following proposition.%

\begin{proposition}
\label{prop:delepsreach_infty}
Suppose $\widehat\M\preceq^{\delta}_\varepsilon\M$ with $\delta>0$, and the specification being expressed as a DFA  $\mathcal{A}$.
Then, for any $\M$ we can construct $\Ca$ such that the specification is satisfied by $\Ca\times\M$ with probability at least $\mathcal S_{\varepsilon,\delta}^\ast$. This quantity is the \emph{($\varepsilon,\delta$)-robust satisfaction probability} defined as
\begin{equation*}
\mathcal S_{\varepsilon,\delta}^\ast
:=\min_{\bar q_0\in \bar\trans(q_0,\hat x_0)}\!\!\max\left(\ind_F(\bar q_0), \Vb_\infty^{\ast}(\hat x_0,\bar q_0)\right),
\end{equation*}
where $\Vb_\infty^{\ast}$ is the unique solution of the fixed-point equation
\begin{equation}
\label{eq:fixed_point}
\Vb_{\infty}^\ast = \mathbf T^\ast_{\eps,\delta} (\Vb_{\infty}^\ast),
\end{equation}
obtained from $\Vb_\infty^\ast :=\lim_{l\rightarrow\infty} (\mathbf T^\ast_{\varepsilon,\delta})^{l}(\Vb_0)$ with $\Vb_0 =0$.
The abstract controller $\Cah$ is the stationary Markov policy $\mu^\ast$
 that maximizes the right-hand side of \eqref{eq:fixed_point}, i.e., $\mu^\ast=\arg \sup_{\mu} \mathbf T_{\eps,\delta}^{\mu} (\Vb_\infty^\ast)$.
The controller $\Ca$ is the refined controller obtained from the abstract controller $\Cah$, the interface function $\InF$, and the state mapping $\Refi$. 
\end{proposition}





\section{Simulation Relations for Linear and Nonlinear Systems}
\label{sec:establish_relation}

In this section, we apply the previously defined concepts
to construct simulation relations and show how to answer Prob.~\ref{prob:prob2} first on a simple linear system, and then on general nonlinear systems.

\subsection{Simulation relations for linear systems}\label{sec:simrel_affine}

Consider the following uncertain linear system
\begin{equation}\label{eq:lti_uncertain}
	\M(\theta): \left\{ \begin{array}{ll}
		x_{k+1} &= Ax_{k}+ B u_k+ w_k+\theta\\
		y_k &= Cx_k,
	\end{array} \right.
\end{equation}
with known matrices $A,B,C$, an uncertain parameter $\theta\in\Theta$, and the nominal model
\begin{equation}\label{eq:lti_nom}
	\Mh: \left\{ \begin{array}{ll}
		\hat x_{k+1} &= A \hat x_{k}+ B \hat u_k+\hat w_k\\
		\hat y_k &= C \hat x_k.
	\end{array} \right.
\end{equation}
The noises $w_k,\hat w_k$ have Gaussian distribution with a full rank covariance matrix. Without loss of generality, we assume that $w_k$ and $\hat w_k$ have the distribution $\mathcal N(0,I)$, with $I$ being the identity matrix.
The probability of transitioning from a state $x_k$ with input $u_k$ to a state $x_{k+1}\in S$ is given by $P(x_{k+1}\in S|x_k,u_k) = \int_{S} \Tr(dx_{k+1}|x_k,u_k)$, where
\begin{align*}
		\Tr(\dxp|x,u;\theta) &= \mathcal N(\dxp|Ax+Bu+\theta, I)\\
		&	= \int_w \fdel{\dxp|w;\theta} \mathcal N(dw|0, I).
\end{align*}
In this equation,  $\mathcal N$ and $f_\delta$ are the normal and Dirac stochastic kernels defined by
\begin{align*}
&\mathcal N(\dxp|m, \Sigma) \!:=\!\frac{\dxp}{ \sqrt{(2\pi)^n \!\left|\Sigma\right| }}\exp\!\left[\!-\frac{1}{2}(x^+\!\!\!-\!m)\T\!\Sigma^{-1}\!(x^+\!\!\!-\!m)\!\right]\nonumber\\
& \fdel{\dxp|w;\theta}:= \delta_{Ax+Bu+\theta+w}(\dxp),
\end{align*}
with $n := \mathrm{dim}\left(\Sigma\right)$ and $|\Sigma|$ the determinant of $\Sigma$.

Similarly, the stochastic kernel of $\Mh$ is
\begin{align*}
		\hat\Tr(\dxhp|\hat x,\hat u)& = \mathcal N(\dxhp|A\hat x+B\hat u, I)\\
		&	= \int_{\hat w} \fhdel{\dxhp|\hat w} \mathcal N(d\hat w|0, I),
\end{align*}
with $\fhdel{\dxhp|\hat w}:=\delta_{A\hat x+B\hat u+\hat w}(\dxhp)$.
Note that $\fdel{\cdot}$ and $\fhdel{\cdot}$ are dependent on $x,u$, and $\hat x,\hat u$, respectively, omitted here for conciseness.

We choose the interface function $\ac{k}\sim\InF(\cdot|\hat x_k, x_k, \hat u_k):=\ach{k}$,
and the noise coupling $\hat w_k\equiv w_k + \theta$ to get the state mapping
\[ \xhp = \xp-A (x_k-\hat x_k)-B (\ac{k}-\ach{k}), \]
not dependent on $\theta$, i.e.,
\begin{equation}
\label{eq:statemap}
	\xhp \sim \Refi(\cdot |\xh{k}, \x{k}, \x{k+1},\ach{k}) := \delta_{\xp-A x^\Delta_k-B u^\Delta_k}(d\xhp),
\end{equation}
where $x^\Delta_k:=x_k-\hat x_k$ and $u^\Delta_k:=\ac{k}-\ach{k}$.
To establish an SSR $\widehat{\M}\preceq^{\delta}_\eps\M(\theta)$, we select the relation
\begin{equation}
\label{eq:relation_linear}
\R:=\{ (\hat x,x)\in\Xh\times\X \given x=\hat x\}.
\end{equation}
Condition (a) of Def.~\ref{def:subsim} holds by setting the initial states $\hat x_{0} = x_{0}$. Condition (c) is satisfied with $\eps=0$ since both systems use the same output mapping. For condition (b), we define the sub-probability coupling $\Wt(\cdot|\theta)$ of $\hat \Tr(\cdot|\hat x,\ach{})$ and $\Tr(\cdot|x,\InF;\theta)$ over $\R$:
\begin{align}
		&\Wt(d\xhp\times d\xp|\theta) = \int_{\hat w}\int_w \fhdel{\dxhp|\hat w} \fdel{\dxp|w;\theta}\nonumber \\
		&\hspace{50pt}\times\delta_{\theta+w}(d\hat w) \min\{\mathcal N(dw|0,I),\mathcal N(dw|-\theta, I)\},
		\label{eq:subprobcoup}
\end{align}
that takes the minimum of two probability measures.
%
%
%

\begin{theorem}
\label{thm:linear}
For the linear models in \eqref{eq:lti_uncertain}-\eqref{eq:lti_nom}, we have that $\widehat{\M}\preceq^{\delta}_\eps\M(\theta)$ with interface function $\ac{k} = \ach{k}$, relation \eqref{eq:relation_linear}, and the sub-probability kernel \eqref{eq:subprobcoup}. The state mapping \eqref{eq:statemap} defines a valid control refinement with $\eps=0$ and
\begin{equation}
\label{eq:global_delta}
	\delta = \sup_{\theta\in\Theta}\left\lbrace  1 - 2\cdot\cdf{-\frac{1}{2}\norm{\theta}}\right\rbrace,
\end{equation}
where $\cdf{\cdot}$ is the cumulative distribution function of a Gaussian distribution, $\cdf{\alpha}:= \int_{-\infty}^\alpha \frac{1}{\sqrt{2\pi}}\exp(-\beta^2/2)d\beta$.
\end{theorem}

\subsection{Simulation relations for nonlinear systems}
Consider the nonlinear system
\begin{equation}\label{eq:nonlinear_sys}
	\M(\theta): \left\{ \begin{array}{ll}
		x_{k+1} &= f(x_k, u_k;\theta)+w_k\\
		y_k &= h(x_k),
	\end{array} \right.
\end{equation}
 with uncertain parameter $\theta\in\Theta$, and the nominal model
\begin{equation}
	\label{eq:nonlinear_nom}
	\Mh: \left\{ \begin{array}{ll}
		\hat x_{k+1} &=   f(\hat x_{k},\hat u_k; \theta_0) + \hat  w_k\\
		\hat y_k &=    h(\hat x_k),
	\end{array} \right.
\end{equation}
with nominal parameter $\theta_0\in\Theta$. 
Assume that $w_k\sim \mathcal N(0,I)$ and $\hat w_k\sim \mathcal N(0,I)$.
%
%
We can rewrite the model in Eq.~\eqref{eq:nonlinear_sys} as
\begin{equation*}
	x_{k+1} =\underbrace{f(x_k, u_k;\theta_0)}_{\text{nominal dynamics}} + \underbrace{ f(x_k, u_k;\theta)-f(x_k, u_k;\theta_0)+w_k}_{\text{disturbance}}.
\end{equation*}
Note that the disturbance part consists of a deviation caused by the unknown parameter $\theta$ and a deviation caused by the noise $w_k$.
Let us assume that we can bound the former as
\begin{equation}
	\|f(x, u;\theta)-f(x, u;\theta_0)\|\leq d(x, u)\quad  \forall x, u,  \theta.
	\label{eq:disturbance_bound}
\end{equation}
Using the interface function $\ac{k} = \ach{k}$ and the noise coupling
\begin{align}
 & \hat w_k \equiv  \gamma(x_k, u_k,\theta)+w_k \,\,\text{ with }\nonumber\\
 & \gamma(x, u,\theta) := f(x, u;\theta)-f(x, u;\theta_0),
\label{eq:gamma}
\end{align}
we get the state mapping
\begin{equation}
\label{eq:statemap_nonlin}
\xhp = \xp - f(x, u;\theta_0) + f(\hat x,\hat u; \theta_0),
\end{equation}
that is not dependent on $\theta$.
%
To establish an SSR $\widehat{\M}\preceq^{\delta}_\eps\M(\theta)$, we select the identity relation \eqref{eq:relation_linear}.
Condition (a) of Def.~\ref{def:subsim} holds by setting the initial states $\hat x_{0} = x_{0}$. Condition (c) is satisfied with $\eps=0$ since both systems use the same output mapping. For condition (b), we define the sub-probability coupling $\Wt(\cdot|\theta)$ over $\R$:
\begin{align}
		&\Wt(d\xhp\times d\xp|\theta) =  \int_{\hat w}\int_w \fhdel{\dxhp|\hat w} \fdel{\dxp|w;\theta}\nonumber \\
		&\hspace{50pt}\times \delta_{\offset+w}(d\hat w) \min\{\mathcal N(dw|0,I),\mathcal N(dw|-\offset, I)\},
		\label{eq:subprobcoup_nonlin} 
\end{align}
that takes the minimum of two probability measures. Note that we dropped the arguments of $\offset$ to lighten the notation.

\begin{theorem}
\label{thm:nonlinear}
For the nonlinear models in \eqref{eq:nonlinear_sys}-\eqref{eq:nonlinear_nom}, we have that $\widehat{\M}\preceq^{\delta}_\eps\M(\theta)$ with interface function $\ac{k} = \ach{k}$, relation \eqref{eq:relation_linear}, and the sub-probability kernel \eqref{eq:subprobcoup_nonlin}. The state mapping \eqref{eq:statemap_nonlin} defines a valid control refinement with $\eps=0$ and
\begin{equation}
\label{eq:statewise_delta}
	\delta(\xh{k},\ach{k}) = \sup_{\theta\in\Theta}\left\lbrace  1 - 2\cdot\cdf{-\frac{1}{2}\norm{\offset(\xh{k},\ach{k},\theta)}}\right\rbrace,
\end{equation}
and $\gamma$ defined in \eqref{eq:gamma}.
\end{theorem}
This theorem gives a $\delta$ that depends on $(\xh{k},\ach{k})$. The supremum of this quantity over $\Xh\times\Uh$ could be obtained to get a constant global $\delta$, which is more conservative.

\section{Case studies}
\label{sec:case_study}
We demonstrate the effectiveness of our approach on a linear system and the nonlinear discrete-time version of the Van der Pol Oscillator.


\subsection{Linear system}
\label{sec:lti_casestudy}
	Consider the linear systems in Eqs.~\eqref{eq:lti_uncertain}-\eqref{eq:lti_nom} with $A:=[0.9,0;0,0.9]$, $B:=[0.7,0;0,0.7]$, and $C:=[1,0;0,1]$. Let $\hat w_k, w_k\sim \mathcal N(0,0.5I)$. Define the state space $\X = [-10,10]^2$, input space $\U = [-1,1]$, and output space $\Y = \X$.
The goal is to compute a controller for the system to avoid region $P_2$ and reach region $P_1$. The specification is $\psi = P_3 \Until P_1$ with $P_3 := \X \setminus P_2$. Here, we consider
$P_1 =[4,10]\times[-4,0]$
and $P_2 =[4,10]\times[0,4]$.
The uncertainty set is $\Theta = \left[-0.090,0.090\right]^2$ with the nominal parameter $\theta_0 = (0,0)$. 

Using the nominal model as $\Mh$ and Eq.~\eqref{eq:global_delta}, we get that $\Mh\preceq^{\delta_1}_{\eps_1}\M(\theta)$ with $\delta_1=0.051$ and $\eps_1 = 0$.
We compute a second abstract model $\Mt$ by discretizing the space of $\Mh$. Then, we use the results of \cite{haesaert2020robust} to get ${\Mt}\preceq^{\delta_2}_{\eps_2}\Mh$ with $\eps_2=0.950$ and $\delta_2=0$. Thus, we have $\Mt\preceq^{\delta}_{\eps}\M(\theta)$ with $\delta = \delta_1+\delta_2$ and $\eps = \eps_1 + \eps_2$.
The probability of satisfying the specification with this $\eps$ and $\delta$ is computed based on Prop.~\ref{prop:delepsreach_infty} and is depicted in Fig.~\ref{fig:val_fcn_lti} as a function of the initial state of $\M$.


\begin{figure}
	\centering
	\includegraphics[width=.9\columnwidth]{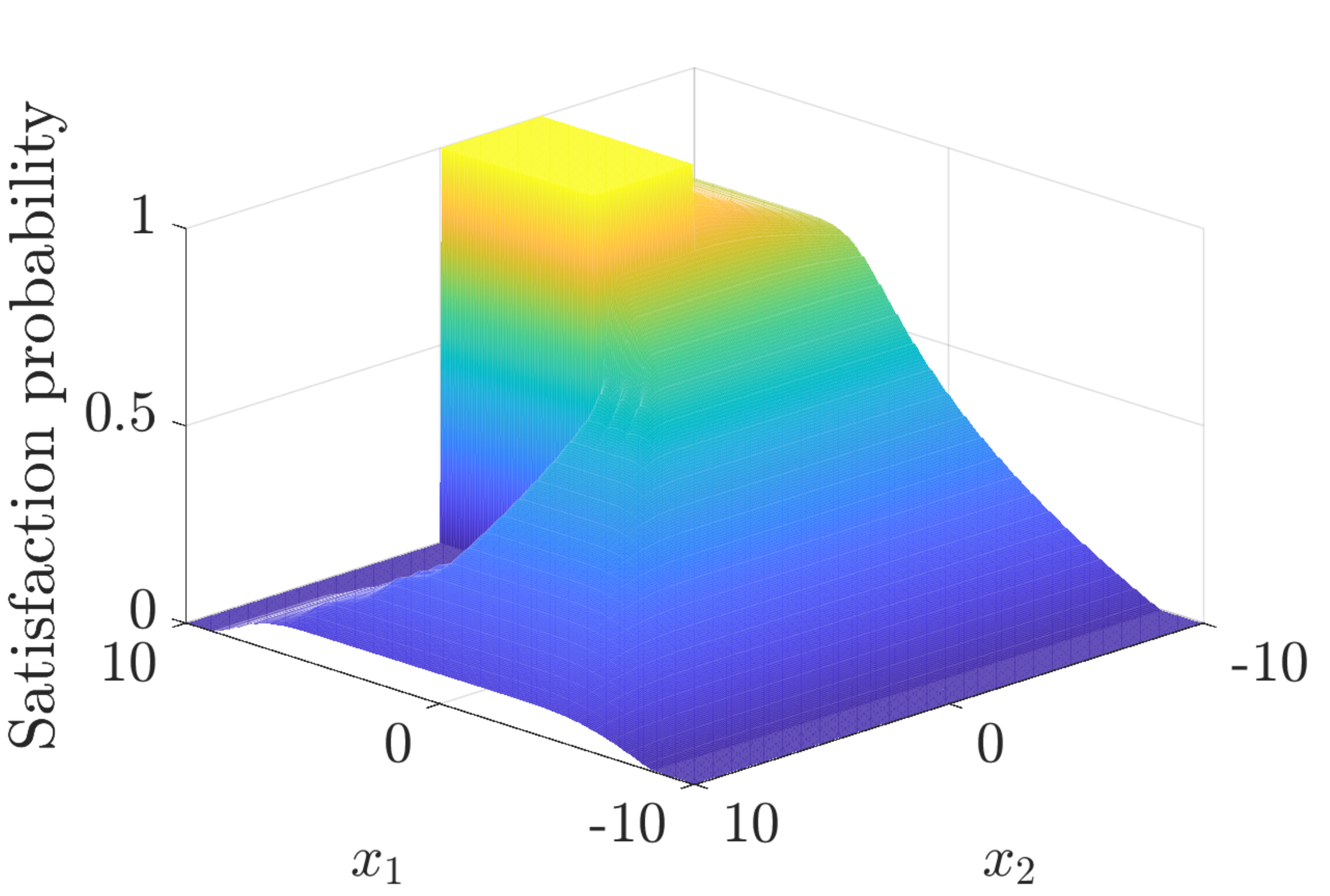}
	\caption{Satisfaction probability for the linear system with $\varepsilon=0.950$, and $\delta=0.051$.
	}
	\label{fig:val_fcn_lti}
\end{figure}

\subsection{Van der Pol Oscillator}\label{sec:vdp}
The state evolution of the Van der Pol Oscillator in discrete time is given as
\begin{align*}
		x_{1,k+1} &= x_{1,k} + x_{2,k} \tau+w_{1,k}
		\\
		x_{2,k+1} & = x_{2,k} + (-x_{1,k}+\theta(1-x_{1,k}^2)x_{2,k}) \tau+ u_k+w_{2,k},
\end{align*}
with sampling time $\tau = 0.1$ and $\hat w_k, w_k\sim \mathcal N(0,0.2I)$.
We consider the state space $\X = \mathbb R^2$, input space $\U = [-1,1]$, and output equation $y_k=x_k$. We want to design a controller such that the system remains inside region $P_1$ while reaching region $P_2$, written as $\psi = P_1 \Until P_2$. The regions are $P_1=[-3,3]^2$ and $P_2 = [2,3]\times[-1,1]$. The uncertainty set is $\Theta = \left[0.700,1.300\right]$ with nominal value $\theta_0=1$.
We select $\Mh$ to be the nominal model and get $\eps_1=0$.
We calculate a state-dependent upper bound on $\offset$ as in Eqs.~\eqref{eq:disturbance_bound}-\eqref{eq:gamma} using $x_k = \hat x_k$,
\begin{equation*}
	d(\xh{k}) = \tau \sup_{\theta\in\Theta} |\theta_0-\theta|\, (1-\xh{1,k}^2)\xh{2,k}.
\end{equation*}
The state-dependent $\delta_1$ is computed using Eq.~\eqref{eq:statewise_delta} and is displayed in Fig.~\ref{fig:deltas} (top). Note that $\delta_1$ grows rapidly for states away from the central region, thus a global upper bound on $\delta_1$
would result in a poor satisfaction probability.
We compute a second abstract model $\Mt$ by discretizing the space of $\Mh$ using the method outlined in \cite{huijgevoort2022piecewiseaffineabstraction}. Then, we use the results therein to get ${\Mt}\preceq^{\delta_2}_{\eps_2}\Mh$ with $\eps_2=0.1$ and $\delta_2(\hat{x})$ given in Fig.~\ref{fig:deltas} (bottom).
Using the transitivity property in Thm.~\ref{thm:transitive}, we have $\Mt\preceq^{\delta}_{\eps}\M(\theta)$ with $\delta = \delta_1+\delta_2$ and $\eps = \eps_1 + \eps_2$.
The probability of satisfying the specification with this $\eps$ and $\delta$ is computed based on Prop.~\ref{prop:delepsreach_infty} and is given in Fig.~\ref{fig:vdPol_SatProb} as a function of the initial state of $\M(\theta)$.

\begin{figure}
	\centering
	\includegraphics[width=0.9\columnwidth]{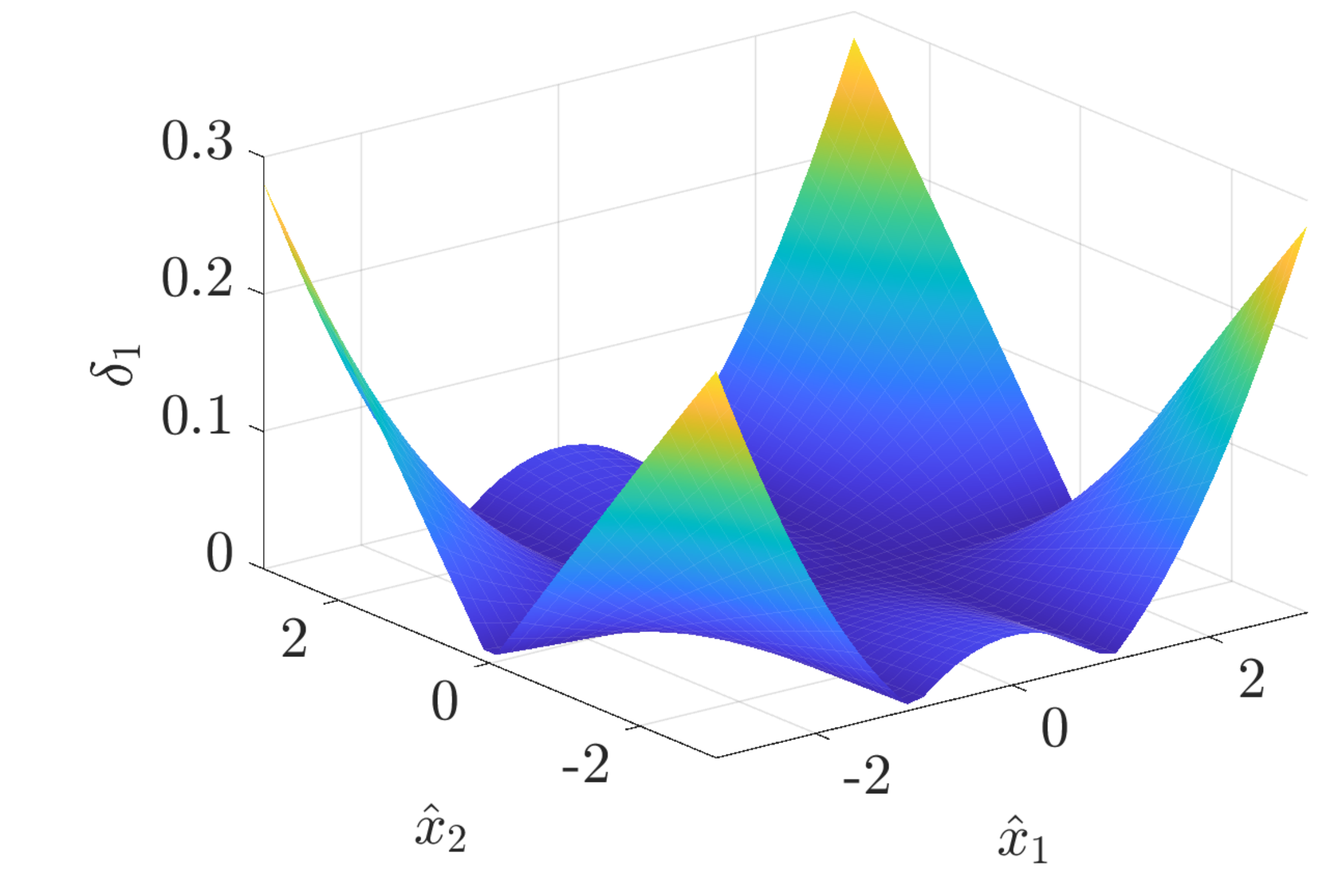}\\
	\vspace{0.1cm}
	\includegraphics[width=0.9\columnwidth]{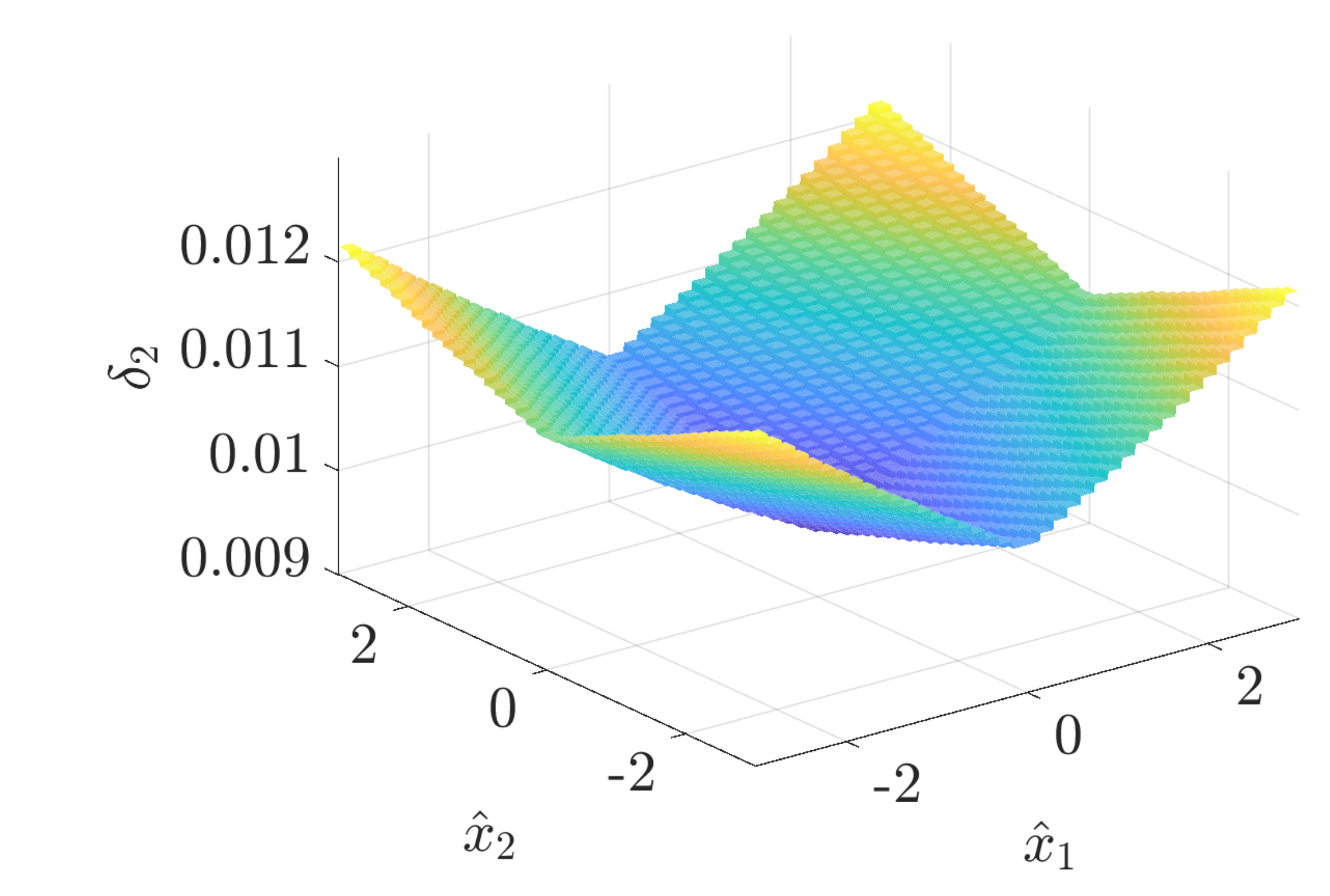}
	\caption{Values of $\delta_1$ (top) and $\delta_2$ (bottom) as a function of the initial state for the Van der Pol Oscillator. 
	}
	\label{fig:deltas}
\end{figure}

\begin{figure}
	\centering
	\includegraphics[width=.9\columnwidth]{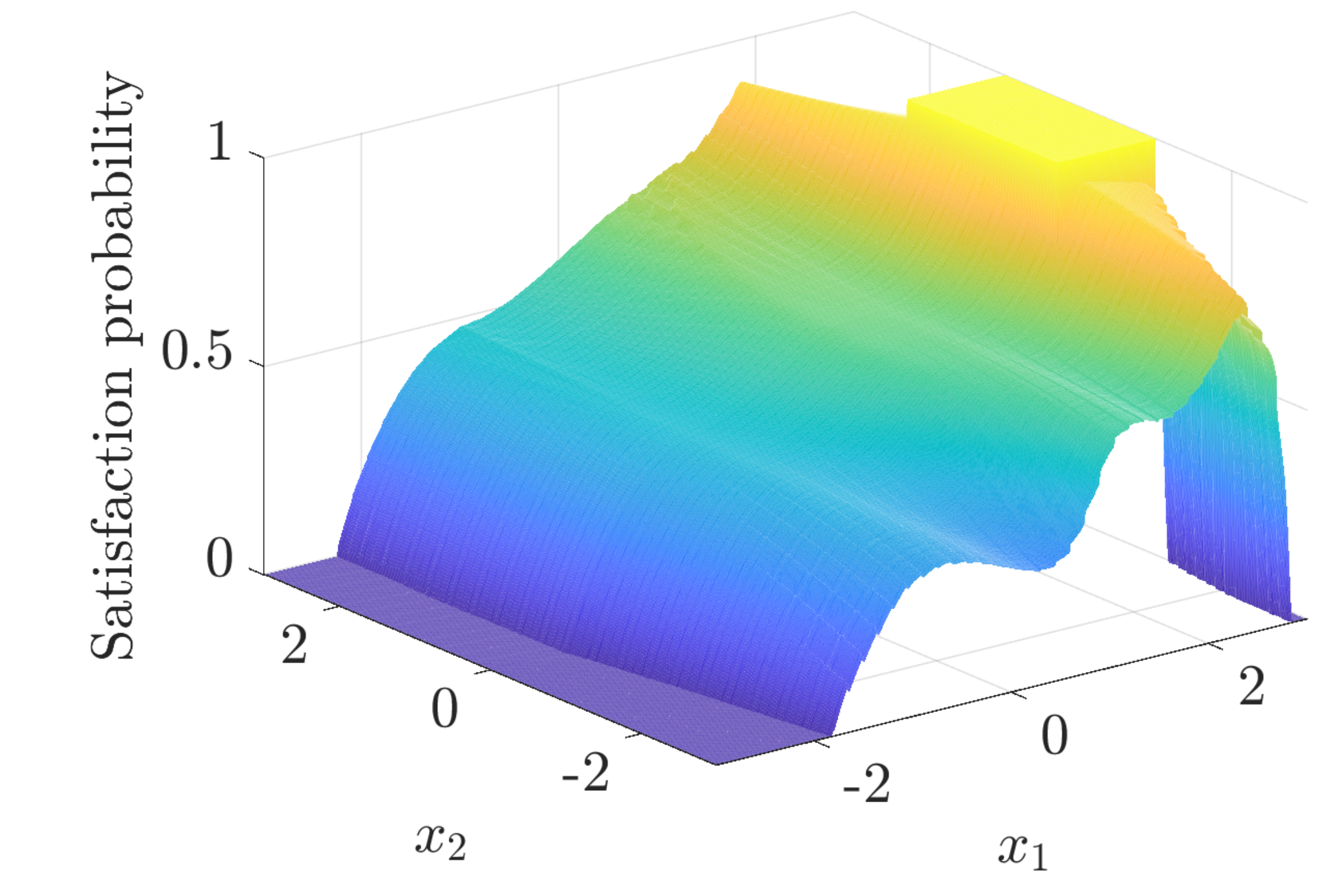}
	\caption{Lower bound on the reachability probability as a function of the initial state for the Van der Pol Oscillator.
	}
	\label{fig:vdPol_SatProb}
\end{figure}

\section{CONCLUSIONS \& FUTURE WORK}
\label{sec:concl}
In this paper, we presented a new similarity relation for stochastic systems that can establish a quantitative relation between a parameterized class of models and a simple abstract model. We showed that this relation can be established on both linear and nonlinear parameterized systems, and provided a method for designing controllers that are robust with respect to errors quantified by the similarity relation to satisfy temporal specifications. In the future, we plan to extend the class of specifications and provide the implementation of the approach for general nonlinear systems.

\bibliographystyle{abbrv}

\bibliography{references.bib}
\end{document}